\documentclass[11pt]{article}

\def\showauthornotes{1}
\usepackage{amsmath}
\usepackage{amssymb}
\usepackage{amsthm}
\usepackage[paper=letterpaper,margin=1in]{geometry}

\usepackage{forloop}
\usepackage{bm}
\usepackage{bbm}
\usepackage{algorithmic}

\usepackage{xspace}

\usepackage[l2tabu, orthodox]{nag}

\usepackage{hyperref}

\usepackage{nth}
\usepackage{color}

\usepackage{array}

\usepackage{paralist}
\usepackage{times}

\usepackage{flafter}

\usepackage{tikz}
\usetikzlibrary{calc}
\usetikzlibrary{matrix}
\usetikzlibrary{shapes}
\usetikzlibrary{decorations.pathmorphing}
\usetikzlibrary{decorations.pathreplacing}

\newtheorem{theorem}{Theorem}[section]
\newtheorem{lemma}[theorem]{Lemma}

\newtheorem{corollary}[theorem]{Corollary}
\newtheorem{observation}[theorem]{Observation}

\newtheorem{claim}[theorem]{Claim}
\newtheorem{fact}[theorem]{Fact}
\newtheorem*{fact*}{Fact}

\theoremstyle{definition}

\makeatletter
\newtheorem*{rep@theorem}{\rep@title}
\newcommand{\newreptheorem}[2]{%
\newenvironment{rep#1}[1]{%
 \def\rep@title{#2 \ref{##1}}%
 \begin{rep@theorem}}%
 {\end{rep@theorem}}}
\makeatother

\newreptheorem{reduction}{Reduction}
\newreptheorem{lemma}{Lemma}

\colorlet{darkgreen}{green!50!black}

\newcommand{\defcal}[1]{\expandafter\newcommand\csname c#1\endcsname{{\mathcal{#1}}}}
\newcounter{ct}
\forLoop{1}{26}{ct}{
    \edef\letter{\Alph{ct}}
    \expandafter\defcal\letter
}

\ifnum\showauthornotes=1
\newcommand{\Authornote}[2]{{\sffamily\small\color{red}{[#1: #2]}}}
\newcommand{\Authornotecolored}[3]{{\sffamily\small\color{#1}{[#2: #3]}}}
\newcommand{\Authorcomment}[2]{{\sffamily\small\color{gray}{[#1: #2]}}}
\newcommand{\Authorstartcomment}[1]{\sffamily\small\color{gray}[#1: }

\newcommand{\Authorfnote}[2]{\footnote{\color{red}{#1: #2}}}
\newcommand{\Authorfixme}[1]{\Authornote{#1}{\textbf{??}}}
\newcommand{\Authormarginmark}[1]{\marginpar{\textcolor{red}{\fbox{\Large #1:!}}}}
\else
\newcommand{\Authornote}[2]{}
\newcommand{\Authornotecolored}[3]{}
\newcommand{\Authorcomment}[2]{}
\newcommand{\Authorstartcomment}[1]{}

\newcommand{\Authorfnote}[2]{}
\newcommand{\Authorfixme}[1]{}
\newcommand{\Authormarginmark}[1]{}
\fi

\DeclareMathOperator{\prefix}{prefix}

\DeclareMathOperator{\Exp}{Exp}

\renewcommand{\vec}[1]{{\bm{#1}}}

\newcommand{\E}{{\mathbb{E}}}

\newcommand{\hide}[1]{}
\newcommand{\sgraph}[1]{\ensuremath{\mathcal{SG}
\ifthenelse{\equal{#1}{}}{}{(#1)}
}}
\newcommand{\cgraph}[1]{\ensuremath{\mathcal{CG}
\ifthenelse{\equal{#1}{}}{}{(#1)}
}}
\newcommand{\cpath}[2]{\ensuremath{P_{#1}
\ifthenelse{\equal{#2}{}}{}{(#2)}
}}
\newcommand{\pfl}{\ensuremath{\mathbb{P}_{\mathsf{FL}}}}

\begin{document}
\title{\Large Dynamic Facility Location via Exponential Clocks\thanks{Supported by ERC Starting Grant 335288-OptApprox.}}
\author{Hyung-Chan An\footnotemark[2]
\and
Ashkan Norouzi-Fard\footnotemark[2]
\and
Ola Svensson\footnotemark[2]}
\date{}

\maketitle

\renewcommand{\thefootnote}{\fnsymbol{footnote}}
\footnotetext[2]{School of Computer and Communication Sciences, EPFL. Emails: \{hyung-chan.an, ashkan.norouzifard, ola.svensson\}\linebreak @epfl.ch}
\renewcommand{\thefootnote}{\arabic{footnote}}

\setcounter{page}{0}
\maketitle
\thispagestyle{empty}

\begin{abstract}
The \emph{dynamic facility location problem} is a generalization of the classic facility location problem proposed by Eisenstat, Mathieu, and Schabanel to model the dynamics of evolving social/infrastructure networks. The generalization lies in that the distance metric between clients and facilities changes over time. This leads to a trade-off between optimizing the classic objective function and the ``stability'' of the solution: there is a switching cost charged every time a client changes the facility to which it is connected. While the standard linear program (LP) relaxation for the classic problem naturally extends to this problem, traditional LP-rounding techniques do not, as they are often sensitive to small changes in the metric resulting in frequent switches.

We present a new LP-rounding algorithm for facility location problems, which yields the first constant approximation algorithm for the dynamic facility location problem. Our algorithm installs competing exponential clocks on the clients and facilities, and connect every client by the path that repeatedly follows the smallest clock in the neighborhood. The use of exponential clocks gives rise to several properties that distinguish our approach from previous LP-roundings for facility location problems. In particular, we use \emph{no clustering} and we allow clients to connect through paths of \emph{arbitrary lengths}. In fact, the clustering-free nature of our algorithm is crucial for applying our LP-rounding approach to the dynamic problem. 
\end{abstract}



\newpage

\section{Introduction}
\label{sec:introduction}

The \emph{facility location problem} is an extensively studied combinatorial optimization problem, which can also be understood as a problem of identifying closely related groups of nodes in networks. In this problem, we are given a single metric on a set of \emph{clients} and \emph{facilities}, where each facility is associated with an \emph{opening cost}; the aim of the problem is to choose a subset of facilities to \emph{open} and connect every client to one of these open facilities while minimizing the solution cost. The solution cost is defined as the sum of the opening costs of the chosen facilities, and the \emph{connection cost} given by the total distance between every client and the facility it is connected to.

The \emph{dynamic facility location problem} is a generalization of this classic problem to temporally evolving metrics, proposed by Eisenstat, Mathieu, and Schabanel~\cite{EMS}. The temporal aspect of the problem is modeled by $T$ metrics given on the same set of clients and facilities, each representing the metric at time step $t\in\{1,\ldots,T\}$. The problem asks us to find a feasible connection of the clients for each time step, but minimizing a new objective function: in addition to the classic opening and connection costs, we incur a fixed amount of \emph{switching cost} every time a client changes its connected facility between two consecutive time steps. This modification was introduced to favor ``stable'' solutions, as Eisenstat et al.\ proposed this problem in order to study the dynamics of evolving systems. Given a temporally changing social/transportation network, the dynamic facility location problem aims at discovering temporal evolution of groups that is not too sensitive to transient changes in the metric. A more comprehensive discussion on the study of the dynamics of evolving networks can be found later in this section. (Also see Eisenstat et al.~\cite{EMS}.)

For the classic facility location problem, a large number of algorithmic techniques (and their combination) have been successfully applied, including LP-rounding~\cite{STA97, CS04, BA10, Li13}, filtering~\cite{LV92B}, primal-dual methods~\cite{JV01}, dual fitting~\cite{JMS02, JMM03}, local search~\cite{AGKMMP, CG}, and greedy improvement~\cite{CG}. LP-rounding approaches, in particular, have their merit that they easily extend to other related problems with similar relaxations. Interestingly, a common algorithmic tool is shared by these traditional LP-roundings: in rounding an LP solution, we cannot afford to open enough number of facilities to ensure that every client can find an open facility among the ones it is fractionally connected to in the LP solution. Hence, these LP-rounding algorithms define a \emph{short} ``fall-back'' path for each client, and guarantee that at least this fall-back path will always lead to an open facility even if the (randomized) rounding fails to give an open facility in the direct neighborhood. The fall-back paths are constructed based on a certain clustering of the LP-solution, where the algorithm opens at least one facility in each cluster. These clustering decisions, unfortunately, are very \emph{sensitive} to small changes in the input; as a result, when these traditional LP-rounding techniques are applied to the dynamic problem, they can generate an excessive number of switches between two consecutive time steps whose LP connection variables are only slightly different.

In this paper, we present a novel LP-rounding approach based on exponential clocks for facility location problems. Exponential clocks were previously used to give a new approximation algorithm for the multiway cut problem by Buchbinder, Naor, and Schwartz~\cite{BNS}. Several interesting properties distinguish our algorithm from previous LP-rounding approaches for facility location problems: firstly, our algorithm allows a client to be connected along an \emph{arbitrarily long path} in the LP support, in contrast to the traditional fall-back paths. While it may appear counter-intuitive at a glance that allowing longer paths helps, use of exponential clock guarantees that the probability that a long path is actually used rapidly diminishes to zero as we consider longer paths. On the other hand, this small probability of using long paths is still sufficient to eliminate the need of fall-back paths, leading to the second property our algorithm: it does not rely on any clustering. Our algorithm, consequently, becomes ``stable'' with respect to small changes in the LP solution. For the dynamic problem, \emph{separately} applying our new LP-rounding for each time step, but with \emph{shared} randomness, ensures that our algorithm makes similar connection decisions for any two time steps whose LP connection variables are similar. Our approach thereby yields the first constant approximation algorithm for the dynamic facility location problem. We also note that our algorithm is a Lagrangian-preserving constant approximation algorithm for the (classic) facility location problem, although with a worse approximation ratio than the smallest known.

Eisenstat et al.~\cite{EMS} proposed $O(\log nT)$-approximation algorithms for the dynamic problem that avoid the stability issue in a different way: they connect every client directly to one of the randomly opened facilities to which the client is fractionally connected in the LP solution, where the random choices are made based on exponential distributions. Such direct connection keeps the algorithm from relying on the triangle inequality, yet any algorithm that does not assume the triangle inequality cannot achieve sublogarithmic approximation under complexity-theoretic assumptions. Eisenstat et al.\ in fact considered two versions of the dynamic facility location problem for both of which they presented logarithimic approximation algorithms: in one version, the facility opening decision is global -- paying the opening cost makes the facility available at \emph{every} time step. In the other version, considered in the present paper, more flexibility is given to the facility opening decision: a facility is opened for a specific set of time steps, and \emph{hourly} opening cost is paid for each time step. They showed that the first version does not admit a $o(\log T)$-approximation algorithm even for the metric case, while leaving it an open question whether a constant approximation algorithm is possible for the second case, which we now positively answer.

\paragraph{Related work.}
A huge amount of data is gathered by observing social networks such as face-to-face contact in a primary school~\cite{SVBCIPQVRLV}, where these networks evolve over time. Different tools have been suggested and analyzed in order to understand the dynamic structure of these networks~\cite{N,TBK,PV}. Dynamic facility location problem is a new tool to analyze temporal aspects of such networks, introduced by Eisenstat et al.~\cite{EMS}. In this paper, we present a constant approximation algorithm for this problem.

Apart from the offline and the dynamic versions of the facility location problem discussed so far, the online setting is a well studied one (see~\cite{F11} for a survey).  In this setting, the clients arrive one at a time and we need to connect them to facilities on the fly. The study of this online setting was started by Meyerson~\cite{M}, who achieved a competitive ratio of $O(\text{log }n)$. Later, Fotakis~\cite{F08} showed an asymptotically tight competitive ratio of~$\Theta(\text{log }n/ \text{log log }n)$. Online problems have also been studied under varying assumptions to give constant competitive algorithms: Anagnostopoulos, Bent, Upfal, and Hentenryck~\cite{ABUV} studied the case where the clients are drawn from a known distribution; Fotakis~\cite{F06} presented an algorithm for the case where the reassignment of clients is allowed; Div\'eki and Imreh~\cite{DI} considered a setting that allows us to move facilities. 

Finally, the facility leasing problem is a variant of the facility location problem introduced by Anthony and Gupta~\cite{LEA1}, who considered a family of leasing problems. While the facility leasing problem also aims at connecting clients to open facilities over multiple time steps, there exist major differences from the dynamic facility location problem, such as the existence of switching costs. Nagarajan and Williamson~\cite{LWILL} presented a $3$-approximation algorithm for the facility leasing problem.

\vspace{-.3ex}
\paragraph{Overview of our approach.}
The standard LP relaxation for the classic problem consists of two types of decision variables: \emph{opening variables} indicating whether each facility is to be opened, and \emph{connection variables} that indicate whether each pair of client and facility is to be connected. Our algorithm for the dynamic problem starts by solving the natural extension of the standard LP~\cite{EMS}, which augments the LP with a new set of \emph{switching variables} that reflect the $\ell_1$-distances between the connection variables of consecutive time steps. Once we obtain an optimal solution, we apply the preprocessing of  Eisenstat et al.\ to ensure that the connection variables of each client does not change too often compared to the switching cost paid: each time a client changes its (fractional) connection variables, at least one half of the switching cost is paid by the LP solution.

The second step of our algorithm is then to install competing exponential clocks, i.e., to sample a value from an exponential distribution, on every client and facility. These exponential clocks are said \emph{competing}, as the random choices made by our algorithm are based on comparing the clocks on subsets of nodes and choosing the best (i.e., smallest-valued) one. After sampling these clocks, our algorithm considers the LP solution of each time step \emph{separately} to construct an assignment for that time step, but based on a single set of exponential clocks shared across all time steps. 

At each time step, in order to determine which facility is to be connected to a given client $j$, our algorithm constructs a path called \emph{connection path}. The path starts from $j$, and iteratively proceeds to the smallest-clock node among the neighborhood of the current node. If this ``smallest-following'' path enters a cycle, we stop and connect $j$ to the last facility seen. Under this random process of connecting $j$, the path may become very \emph{long}, which is one of the unusual characteristics of our algorithm discussed earlier. However, observe from the construction that the sequence of clock values that this smallest-following path witnesses keep decreasing: in other words, in order for this path to grow, it has to see in the neighborhood a clock that beats everything seen so far. As such, as the path becomes longer, the probability that the path continues will rapidly diminish (Lemma~\ref{lem:pathprob}). For any given edge $e$, this key observation implies that most of the paths that start from a distant node will fail to reach $e$; in fact, we show via a counting argument that the expected number of connection paths passing through $e$ is within a constant factor of its connection variable (Corollary~\ref{cor:marginal}). This bounds the connection cost (and opening cost, with some additional arguments: see Lemma~\ref{lem:opencost}) within a constant factor of the LP cost.

Finally, in order to bound the switching cost, recall that the exponential clocks are shared by all time steps. Hence, if the LP solution did not change at all between two time steps, the random construction of connection paths would not change, either. Now, if the LP solution did change slightly, for example around a single client, an obstacle in the analysis would be that the change in the single client's connection may lead to the switches of multiple clients' connections, whereas the LP solution only pays a constant fraction (namely $1/2$) of a single switch cost. Recall that, however, the connection paths tend to be short; thus, a local change in the connection variable cannot affect the connection of a distant client with high probability, and indeed we show that a change in a single client's connection globally causes only constant number of switches in expectation (Lemma~\ref{lem:switch1}). This yields the last piece of analysis to establish that our algorithm is a constant approximation algorithm.

\vspace{2ex}

Our new LP-rounding approach for facility location problems raises several
interesting research directions. In general, we envision that a further understanding of 
the techniques using exponential clocks will be fruitful for these problems. A more specific question is  
how far the approximation guarantee of our current analysis can be pushed. We know that, by
incorporating more case analyses, the guarantee on the connection costs can be improved. However, it remains an interesting open problem to
understand if a different analysis can lead to bounds that compete
with the best known ratios for the classic facility location problem. Further, while the use of connection paths is important for the
stability of the solution, a potential improvement of  our algorithm when applied to  the classic facility
location problem is in connecting each client to the closest opened facility instead of always following the connection
path.   


\section{Preliminaries}
\label{sec:preliminaries}

\subsection{Facility Location in Evolving Metrics}
\paragraph{Problem definition.}
In dynamic facility location, we are given a set of facilities $F$, clients $C$, and a temporally changing metric on them. We denote by $d_t(i,j)$ the metric distance between client $j$ and
facility $i$ at time $t$. We are also given a switching cost $g$, the total number of time steps $T$, and an opening cost $f_i$ for each facility $i$.
The goal is to output, for each time step $t$, a
subset of open facilities $A_t$ and an assignment $\phi_t:C\rightarrow A_t$ of clients to facilities
so as to minimize:
\begin{align}
\sum_{1\leq t\leq T, i\in A_t} f_i + \sum_{1\leq t\leq T, j\in C} d_t(\phi_{t}(j),j)+\sum_{1\leq t <
  T, j\in C} \mathbbm{1}\{\phi_{t}(j) \neq \phi_{t+1}(j)\}\cdot g,
\end{align}
where $\mathbbm{1}\{p\}$ is the indicator function of proposition $p$, i.e., it takes value $1$ if $p$ is true
and $0$ otherwise. 
In words, the objective function consists of  the hourly opening costs for each open facility, the
connection costs of
each client, and the switching costs. 
\vspace{-1ex}
\paragraph{Linear programming relaxation.}
We first introduce the standard  linear programming relaxation for the classic
facility location problem (or, equivalently, the dynamic version with a single time step). We then formulate
the relaxation for the dynamic facility location problem, introduced in~\cite{EMS}, which is a
natural generalization of the relaxation for the classic facility location problem.

In the standard LP-relaxation of the classic facility location problem, we have a variable $y_i$
for each facility $i\in F$ and a variable $x_{ij}$ for each facility $i\in F$ and client $j\in
C$. The intuition of these variables is that $y_i$ should take value $1$ if $i$ is opened and $0$
otherwise; $x_{ij}$ should take value $1$ if client $j$ is connected to facility $i$ and $0$
otherwise.  The set of feasible solutions to the relaxation is now described by $\pfl =
  \{(x,y) \mid \sum_{i\in F} x_{ij} = 1,\  \forall j \in C;\; x_{ij} \leq
  y_i, \   \forall i \in F, j\in C;\mbox{ and } x,y \geq \vec{0}\}$. The first set of inequalities says that each client should be
  connected to a facility and the second set says that if a client $j$ is connected to a facility
  $i$ then that facility should be opened.  In this terminology, the standard LP relaxation of the classic facility
  location problem is the following:
\begin{align*}
\textrm{minimize } &  \sum_{i\in F} y_i f_i + \sum_{i\in F,
    j\in C} x_{ij} d(i,j) \\
\textrm{subject to } & (x,y) \in \pfl.
\end{align*}

We now adapt the above relaxation to the dynamic facility location problem. Let $[T] = \{1, \ldots,
T\}$ and $[T) = \{1, \ldots, T-1\}$. For each time step $t\in
[T]$, the relaxation has a variable $y_i^t$ for each facility $i\in F$ and a variable $x_{ij}^t$ for
each facility $i\in F$ and client $j\in C$. These variables indicate which facilities should be
opened at time $t$ and where to connect clients at this time step. In other words, 
$(x^t, y^t)$ should be a solution to the classic facility location problem and our relaxation will constrain that
$(x^t, y^t) \in \pfl$ for each $t\in [T]$. To take into account the switching costs, our
relaxation will also have a non-negative variable $z_{ij}^t$ for each client $j\in C$, facility
$i\in F$ and time $t\in [T]$. The intuition of $z_{ij}^t$ is that it should take value $1$ if client
$j$ was connected to facility $i$ at time $t$ but not at time $t+1$. The relaxation of the dynamic
facility location problem introduced in~\cite{EMS} is
\vspace{-0mm}
\begin{align*}
\textrm{minimize } &  \sum_{i\in F, t\in [T]} y^t_i f_i + \sum_{i\in F,
    j\in C, t\in [T]} x^t_{ij} d_t(i,j) + \sum_{i \in F, j\in C, t\in[T)}\ z^{t}_{ij}\cdot g \hspace{-12em}&&\\[1mm]
\textrm{subject to } & (x^t,y^t) \in \pfl  &&\hspace{-5em} \forall  t\in [T], \\
& z_{ij}^t \geq x_{ij}^t - x_{ij}^{t+1} \mbox{ and } z_{ij}^t \geq 0,  &&\hspace{-5em}  \forall  i\in F, j\in C, t\in [T). \\
\end{align*}

\vspace{-9mm} 
\subsection{Exponential Clocks}
We refer to independent exponential random variables as exponential clocks. The
probability density function of an exponential distribution with rate parameter $\lambda >0$ is $f_{\lambda}(x) = \lambda e^{-\lambda
  x}$ for  $x \geq 0$. If a random variable $X$ has
this distribution, we write $X\sim \Exp(\lambda)$. We shall use the  following well-known properties
of the exponential distribution:
\begin{enumerate}\itemsep0mm
\item If $X \sim \Exp(\lambda)$, then $\frac{X}{c} \sim \Exp(\lambda c)$ for any $c>0$.
\item Let $X_1, X_2, ..., X_n$ be independent exponential clocks with rate parameters $\lambda_1, \lambda_2, ..., \lambda_n$, then
\begin{enumerate}
\item $\min\{X_1,...,X_n\} \sim \Exp(\lambda_1+...+\lambda_n)$.
\item\label{clocks2a} $\Pr[X_i = \min\{X_1, ..., X_n\}] = \frac{\lambda_i}{\lambda_1+...+\lambda_n}$.
\end{enumerate}
\item Exponential clocks are memoryless, that is, if $X \sim \Exp(\lambda)$, for any $n,m > 0$:
$$\Pr(X > m+n | X > m) = \Pr(X > n).$$
\end{enumerate}

\vspace{-3mm}

Note that the memorylessness property implies that if we have a set of exponential clocks then after
observing the minimum of value say $v$, the remaning clocks, subtracted by $v$, are still exponentially
distributed  with their original rates.

\subsection{Preprocessing}

The first preprocessing is from the $O(\log nT)$-approximation algorithm by Eisenstat et al.~\cite{EMS}. Losing a factor of
$2$ in the cost, this simple preprocessing lets us assume that the LP pays one switching cost each time a client changes its fractional
connection.
\begin{lemma}[\cite{EMS}]
\label{lem:prepro}
  Given an LP solution, we can, by increasing its cost by at most a factor of $2$, obtain  in polynomial time
  a feasible solution $(x,y,z)$ satisfying:
\begin{itemize}
\item If we let $Z^t = \{j\in C \mid x^t_{ij} \neq
  x^{t+1}_{ij}  \mbox{ for some } i\in F\} $ denote the set of clients that changed its fractional connection between time step $t$
  and $t+1$, then  $\sum_{t=1}^{T-1} |Z^t| \leq \sum_{t=1}^{T-1}\sum_{i\in F, j\in C}z_{ij}^t.$
\end{itemize}

\end{lemma}

The second preprocessing is obtained by using the standard trick of duplicating facilities, while
being careful that if the connection variables of a client remain the same between two consecutive time steps, they remain so even after the preprocessing.

\begin{observation}
\label{obs:prepro}
Without loss of generality, we may assume that $(x,y,z)$ satisfies the following:
\begin{enumerate}
\item For any facility $i\in F$, client $j\in C$, and time step $t\in [T]$, $x_{ij}^t \in \{0,
  y_i^t\}$.\label{enum:prepro2:1}
\item For each facility $i\in F$, there exists $o_i \in [0,1]$ such that $y_i^t \in \{0,o_i\}$ for each time step $t\in
  [T]$.\label{enum:prepro2:2}
\end{enumerate}
\end{observation}

For formal proofs of the above statements, we refer the reader to Appendix~\ref{ap:pre}.


\section{Description of Algorithm}	
\label{sec:algorithm}

Given a preprocessed solution $(x,y,z)$ to the linear programming relaxation that satisfies
the properties of Lemma~\ref{lem:prepro} and Observation~\ref{obs:prepro}, our algorithm proceeds by first making a random choice and
then opening facilities and connecting clients in each time step.
\begin{description}
\item[Random choice:] Sample independently an exponential clock $Q_i \sim \Exp(o_i)$ for each facility $i\in F$
  and an exponential clock $R_j\sim \Exp(1)$ for each client $j \in C$.

\item[Opening and connecting:] At each time step $t\in [T]$, open facilities and connect clients as follows. Consider
  the clients in the non-decreasing order of their sampled clocks ($R_j$'s). When
  client $j\in C$ is considered, find the facility $i = \arg\min_{i: x^t_{ij}>0} Q_{i}$
  of the smallest clock among the facilities that $j$  is connected to in the support of $x^t$. Similarly,
  let $j' =
  \arg\min_{j': x^t_{ij'} > 0} R_{j'}$ be the client with the smallest clock in the neighborhood of $i$. The
  connection of $j$ at time $t$ is now determined as follows:
if $j = j'$ then open $i$ and connect $j$ to $i$; otherwise, connect $j$ to the same facility as $j'$.

\end{description}

Note that, with probability $1$, all clocks will be of distinct values so we make that simplifying assumption.
We also remark that the procedure is well-defined. Indeed, if $j$ connects to the same facility as
$j'$ then $R_{j'} < R_j$ (since both $j$ and $j'$ are adjacent to $i$) , and therefore $j'$ was already
connected to a facility when $j$ was considered.
\subsection{An alternative presentation of our algorithm}
In this section, we rewrite the \emph{opening and connecting} step of the algorithm using graph terminology, simplifying the presentation of our analysis. The new \emph{opening and connecting} step for time step $t$ reads as follows.

Let \sgraph{x^t} be the \emph{support graph} of $x^t$: i.e., \sgraph{x^t} is an undirected bipartite graph on vertex set $F^t\cup C$ that has an edge $\{i,j\}$ if and only if $x^t_{ij}>0$, where $F^t:=\{i\in F\mid \exists j\ x^t_{ij}>0\}$. Then we construct a directed bipartite graph \cgraph{x^t}, called \emph{connection graph}, which represents the random choices made by the algorithm. \cgraph{x^t} is a directed graph on the same vertex set $F^t\cup C$, where every vertex has exactly one outgoing arc directed towards the vertex with the smallest clock among its neighborhood in \sgraph{x^t}. Note that the underlying undirected graph of \cgraph{x^t} is therefore a subgraph of \sgraph{x^t}. Now the algorithm finds all length-2 cycles in \cgraph{x^t}, and opens every facility that appears on any of these cycles.

Once the algorithm determines the set of facilities to be opened, it produces an assignment of the clients to the open facilities. For each client $j$, the algorithm defines its \emph{connection path} \cpath{j}{x^t} as follows: the path starts from $j$, and follows the unique outgoing arc in the connection graph until it stops just before it is about to visit a vertex that it has already visited. Note that this path is well-defined, since every vertex has exactly one outgoing arc in the connection graph and there are finitely many vertices in the graph. The algorithm assigns $j$ to the facility that appears latest on \cpath{j}{x^t}: if \cpath{j}{x^t} ends at a facility, $j$ is assigned to that facility; if \cpath{j}{x^t} ends at a client, the second-to-last vertex of the path is a facility, and $j$ is assigned to that facility. Figure~\ref{f:alg} illustrates an execution of our algorithm.
\begin{figure}[ht]
\input{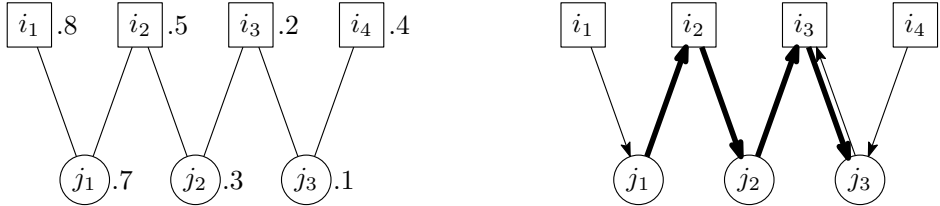}
\caption{An example of the execution of our algorithm. A support graph is shown on the left, where squares represent facilities and circles clients. Each node is annotated with its exponential clock value. The corresponding connection graph is shown on the right, where the bold edges represent $P_{j_1}(x^t)$. Our algorithm connects $j_1$ to $i_3$ in this case.}
\label{f:alg}
\end{figure}

Following is a useful observation in the analysis of our algorithm.
\begin{observation}\label{obs:2cycle}
\cgraph{x^t} does not contain a simple cycle with more than two arcs.
\end{observation}
\begin{proof}
Suppose that \cgraph{x^t} contains a simple cycle $\langle i_1,j_1,i_2,\ldots,i_k,j_k(,i_1)\rangle$ for $k\geq 2$. Note that $i_1\neq i_2$, $\{j_1,i_1\},\{j_1,i_2\}\in\sgraph{x^t}$, and $(j_1,i_2)\in\cgraph{x^t}$; hence, $Q_{i_2}<Q_{i_1}$. Repeating this argument yields $Q_{i_1}<Q_{i_1}$, obtaining contradiction.
\end{proof}
\noindent Each connection path therefore ends only when it reaches a length-2 cycle, and this is why the new presentation of the algorithm is guaranteed to assign every client to an open facility. Observation~\ref{obs:eq} easily follows from the fact that the connection graph is merely a graph representation reflecting the random choices made by the original algorithm.

\begin{observation}\label{obs:eq}
The two versions of our algorithm are equivalent: they open the same set of facilities and produce the same assignment.
\end{observation}

\begin{proof}
First we verify that both versions of our algorithm open the same set of facilities. As the connection graph is bipartite, every length-2 cycle contains exactly one facility and one client. Let $\{(i,j),(j,i)\}$ be a length-2 cycle where $i\in F$ and $j\in C$; the new version opens $i$ in this case. When the original version considers $j$'s connection, it selects $i$ in the neighborhood and subsequently finds that $j$ has the smallest clock in the neighborhood of $i$. Hence the original version also opens $i$.

Note that this in fact is the only case in which the original version opens a facility: $i'\in F$ is opened if and only if, for some client $j'$, $i'$ had the smallest clock in the neighborhood of $j'$ and vice versa. This emerges as a length-2 cycle in the connection graph, and therefore the new version also opens $i'$.

Now we show by induction that, for all $k\in\mathbb{N}$, a client $j$ whose connection path consists of $k$ arcs is assigned to the same facility by both versions of our algorithm. Suppose $k=1$. In this case, $\cpath{j}{x^t}=\{(j,i)\}$ and $(i,j)\in\cgraph{x^t}$. Thus, the original version opens $i$ and assigns $j$ to $i$, which is consistent with the decision made by the new version. Suppose $k=2$. In this case, $\cpath{j}{x^t}=\{(j,i),(i,j')\}$ and $(j',i)\in\cgraph{x^t}$. The new version assigns $j$ to $i$; the original version assigns $j$ to the same facility as $j'$, and assigns $j'$ to facility $i$. Hence, the decisions are consistent in this case as well.

Suppose that $k$ is an odd integer greater than two, and let $(j,i_1),(i_1,j_1)$ be the first two arcs of \cpath{j}{x^t} and $(j_2,i_2)$ be the last. Then the outgoing arc from $i_2$ in the connection graph has to be towards $j_2$, since it has to be towards a vertex that is already visited, creating a length-2 cycle. (See Observation~\ref{obs:2cycle}.) Thus, we can obtain \cpath{j_1}{x^t} by removing the first two arcs from \cpath{j}{x^t}, and the original version of the algorithm assigns $j_1$ to $i_2$ from the induction hypothesis. It assigns $j$ also to $i_2$, which is consistent with the decision made by the new version of our algorithm.

The final case where $k$ is even follows from a symmetric argument.
\end{proof}

\section{Analysis}
\label{sec:analysis}

In this section we analyze our algorithm described in Section~\ref{sec:algorithm} for the
dynamic facility location problem. Throughout this section,  let $X_{ij}^t$ denote the
random indicator variable that takes value $1$ if the algorithm connects client $j$ to facility $i$ at time step $t$
and let $Y_i^t$ be the random indicator variable that takes value $1$ if facility $i$ is opened during time
step $t$.  All probabilities and expectations in this section are  over the random
outcomes of the
exponential clocks.
With this notation and by linearity of expectation, the expected cost of the returned solution equals
\begin{align*}
\sum_{t=1}^T \underbrace{\E\left[\sum_{i\in F} f_i \cdot Y_i^t + \sum_{i\in F, j\in C} d_t(i,j)
  X_{ij}^t \right]}_{(i)}   
+
g \cdot \sum_{t=1}^{T-1} \underbrace{\E \left[\sum_{j\in C} \mathbbm{1}\{X_{ij}^t \neq X_{ij}^{t+1}\mbox{ for
    some } i\in F\}\right]}_{(ii)},
\end{align*} 
where the term~(i) expresses the expected opening cost and connection cost at time step $t$ and the
second term~(ii) expresses the expected number of clients who changed, between time steps $t$ and $t+1$, the facility to which they are
connected.  Note that analyzing $(i)$ is simply the problem of
analyzing our algorithm for the uncapacitated facility location problem. To analyze (ii) we
crucially rely on the fact that the random choices of our algorithm (i.e., the sampling of the exponential
clocks) are based on a single set of exponential clocks shared by all time steps. This will allow us
to prove that the expected number of clients that
change connection is proportional to the cost that the LP pays towards the switching cost.

More specifically, we prove the following lemma:
\begin{lemma}
\label{lem:inequalities}
For any  time step $t\in [T]$, we have
\begin{align}
\label{eq:opencost} \E\left[\sum_{i\in F} f_i Y_i^t\right] & \leq \sum_{i\in F} f_i y_i^t,  \\
\label{eq:conncost} \E\left[\sum_{i\in F, j\in C} d_t(i,j) X^t_{ij} \right] & \leq 6 \sum_{i\in F, j\in C} d_t(i,j)
x_{ij}^t, \\ 
\label{eq:switchcost} g\cdot \E \left[\sum_{j\in C} \mathbbm{1}\{X_{ij}^t \neq X_{ij}^{t+1}\mbox{ for
    some } i\in F\}\right] & \leq 7 g  |Z^t|.
\end{align}
\end{lemma}
The above lemma implies that our algorithm is a $14$-approximation algorithm for the
dynamic facility location, from the following argument. The preprocessing of Lemma~\ref{lem:prepro} incurs a factor
of $2$. Combining this with the above lemma gives us that  the opening costs are approximated within a factor of $2$, the
connection cost within a factor of $12$, and the switching cost within a factor $14$ since
$\sum_{t\in [T)} |Z^t| \leq \sum_{i\in F, j\in C, t\in [T)} z^t_{ij}$. Hence, we have the following\footnote{
As can be seen from the argument, we could also have \emph{temporally changing opening cost} without changing our algorithm: we can allow the input to specify
$f_i^t$, the opening cost of facility $i$ at time step $t$.
}:

\begin{theorem}
There is a randomized $14$-approximation algorithm for the dynamic facility location problem.
\end{theorem}

We prove Inequalities~\eqref{eq:opencost},~\eqref{eq:conncost}, and~\eqref{eq:switchcost} of Lemma~\ref{lem:inequalities} in
Sections~\ref{sec:opencost},~\ref{sec:conncost}, and~\ref{sec:switchcost}, respectively.

\subsection{Bounding the opening cost}
\label{sec:opencost}
We show that the probability to open a facility $i$ at time step $t$ equals $y_i^t$.
\begin{lemma}\label{lem:opencost}
For any time step $t\in [T]$ and facility $i\in F$, $\E[Y_i^t] \leq y_i^t$.
\end{lemma}
\begin{proof}

If $i\notin F^t$, it does not appear in \sgraph{x^t} and cannot be opened: $\E[Y_i^t]=0\leq y_i^t$. Suppose $i\in F^t$, and let $j$ be the facility to which $i$ has its outgoing arc in the connection graph: i.e., $(i,j)\in\cgraph{x^t}$. Let $F(j)$ be the set of facilities that are adjacent to $j$ in \sgraph{x^t}. Observe that $i$ will be opened if and only if $(j,i)\in\cgraph{x^t}$, or, in other words, $Q_i = \min_{i'\in F(j)} Q_{i'}$. Note that this event is independent from the event $(i,j)\in\cgraph{x^t}$, since the facility-clocks are independent from the client-clocks. Thus, from Property~\ref{clocks2a} of exponential clocks, we have $\E[Y_i^t]=\frac{y^t_i} {\sum_{i' \in F(j)} y^t_{i'} } = y^t_i,$ where the last equality follows from $1=\sum_{i'\in F(j)} x^t_{i'j} = \sum_{i' \in F_j} y^t_{i'}$. (See
  Observation~\ref{obs:prepro}.)

\end{proof}

\subsection{Bounding the connection cost}
\label{sec:conncost}
In this section we bound the connection cost for a fixed time step $t\in [T]$, i.e.,
$\E\left[\sum_{i\in F, j\in C} d_t(i,j) X^t_{ij} \right]$. 
 As the time step $t$ is fixed
throughout this section, we
simplify the notation by letting $(x,y)= (x^t, y^t)$ and we also abbreviate
$\sgraph{x^t} ,\cgraph{x^t},$ and $\cpath{j}{x^t}$ by $\sgraph{}, \cgraph{}$, and $\cpath{j}{}$,
respectively.

By the triangle inequality, the connection cost of a client is at most the  sum of distances of the edges in its connection
path. Therefore we have that 
$$
\E\left[\sum_{i\in F, j\in C} d_t(i,j) X^t_{ij} \right] \leq \E\left[ \sum_{j\in C} d_t(\cpath{j}{})\right],
$$
where $d_t(\cpath{j}{})$ denotes the total distance of the edges in the path $\cpath{j}{}$. Note that the right-hand-side can
be further rewritten by summing over all the edges in $\sgraph{}$ and counting the expected number
of connection paths that use this edge (note that the paths do not repeat a vertex). That is, we obtain the following bound on the connection
cost
$$
\sum_{\{i,j\} \in \sgraph{}} d_t(i,j) \E[| \{ j' \in C \mid  (i,j) \mbox{ or } (j,i) \mbox{
   is in } \cpath{j'}{}\}|].
$$

To analyze this, we first  bound the probability, over the randomness of the exponential clocks,
that a  connection path starts with a given prefix. We then show that the expected number of connection paths that traverses an edge $\{i,j\}$ is at most  $6x_{ij}$, which then implies Inequality~\eqref{eq:conncost} of Lemma~\ref{lem:inequalities}.

\subsubsection{The probability of a prefix}

Consider a client $j_0$  and  its connection path $\cpath{j_0}{}$. We use the notation
$\prefix(\cpath{j_0}{}) $ to denote the set of all the prefixes of this path, i.e., the subpaths of
$\cpath{j_0}{}$ that start at $j_0$.
We also let $C(i)$ denote  the set of clients adjacent  to a facility $i$ in  $\sgraph{}$. Similarly, let
$F(j)$ denote the set of facilities adjacent to client $j$ in $\sgraph{}$. We further abbreviate
$C(i_1) \cup \cdots \cup C(i_\ell)$ by $C(i_1, \ldots, i_\ell)$ and $F(j_1) \cup \cdots \cup
F(j_\ell)$ by $F(j_1, \ldots, j_\ell)$. Finally, for a subset $F' \subseteq F$ of facilities, we let
$y(F') = \sum_{i\in F'} y_i$. Using this notation, we now bound the probability that a given subpath appears
as a  prefix
 of a connection path.

\begin{lemma}
\label{lem:pathprob} 
We have 
\begin{align*}
\Pr[\langle j_0,i_1, j_1, i_2, \ldots ,i_k, j_k, i_{k+1}\rangle   \in \prefix(\cpath{j_0}{})] & \leq  \prod_{\ell=1}^k\frac{1}{|C(i_1, i_2, \ldots, i_\ell)|} \cdot  \prod_{\ell=0}^k\frac{y_{i_{\ell+1}}}{y(F(j_0,
    j_1, \ldots, j_\ell)) }, \\
\Pr[\langle j_0,i_1, j_1, i_2, \ldots ,i_k, j_k\rangle   \in  \prefix(\cpath{j_0}{})] &\leq  \prod_{\ell=1}^k\frac{1}{|C(i_1, i_2, \ldots, i_\ell)|} \cdot  \prod_{\ell=0}^{k-1}\frac{y_{i_{\ell+1}}}{y(F(j_0,
    j_1, \ldots, j_\ell)) }.
\end{align*}
\end{lemma}
\begin{proof}
We start by analyzing  $\Pr[\langle j_0,i_1, j_1, i_2, \ldots ,i_k, j_k, i_{k+1}\rangle \in
\prefix(\cpath{j_0}{})]$.  First note that $\langle j_0,i_1$, $ j_1, i_2, \ldots ,i_k, j_k, i_{k+1}\rangle \in
\prefix(\cpath{j_0}{})$ if and only if all the arcs $(j_0, i_1), (i_1, j_1), \ldots, (j_k,
i_{k+1})$ exist in $\cgraph{}$. The algorithm uses two independent sources of randomness:
the client-clocks and facility-clocks.  We use the randomness of the client-clocks to bound the
probability that all the arcs $(i_1, j_1),(i_2, j_2), \ldots, (i_k, j_k)$ exist in $\cgraph{}$. Note that
for $(i_\ell, j_\ell)$ to exist, $j_\ell$ has to have the smallest clock of the clients in
$C(i_\ell)$.  Moreover, as
$j_{\ell-1} \in C(i_\ell)$, we have $R_{j_\ell} < R_{j_{\ell-1}}$. By repeating this argument, we
have that a necessary condition for all arcs to exist in $\cgraph{}$ is that
$R_{j_\ell} < R_{j_{\ell-1}} < \ldots < R_{j_1}$, which implies that the arcs exist only if
$$
R_{j_{\ell}} = \min \{R_j \mid j\in C(i_1, i_2, \ldots, i_\ell)\} \qquad \mbox{ for } \ell = 1,
\ldots, k.
$$
To bound the probability that these conditions hold, we shall use the well-known properties of the
exponential distribution. Note first that all client-clocks are distributed according to the
exponential distribution with the same rate (which is $1$) and therefore  the probability of $R_{j_k} = \min \{R_j \mid j\in
C(i_1, i_2, \ldots, i_k)\}$ is 
\begin{align}
\label{eq:facedge}
\frac{1}{|C(i_1, i_2, \ldots, i_k)|}.
\end{align}
 Now, by the memorylessness
property, if we condition on the event that $j_k$ has the smallest exponential clock of all clients
in $C(i_1, i_2, \ldots, i_k)$, then the clocks  of the clients different from $j_k$, subtracted by $R_{j_k}$,
are still distributed
according to the exponential distribution with same rates. Therefore, the probability of
$R_{j_{k-1}} = \min \{R_{k-1} \mid j\in C(i_1, i_2, \ldots, i_{k-1})\}$  is $1/|C(i_1, i_2, \ldots,
i_{k-1})|$ even if we condition on the
event $ R_{j_k} = \min \{R_j \mid j\in C(i_1, i_2, \ldots, i_{k})\}$.\footnote{Since the exponential clocks have all the same rate an equivalent more
  combinatorial point of view is the following: choose a random permutation of the clients in
  $C(i_1, i_2, \ldots, i_k)$. The probability that a certain client is the first (smallest clock) is
  $1$ over the cardinality of the set; and even after conditioning on this event all permutations of
  the remaining clients are equally likely.} Note that $j_k \not \in C(i_1, \ldots, i_{k-1})$ and
more generally $j_{\ell} \notin C(i_1, \ldots, i_{\ell-1})$ (otherwise,  as $R_{j_\ell} <
R_{j_{\ell-1}}< \ldots < R_{j_1}$,
one of the facilities in $\{i_1, \ldots, i_{\ell-1}\}$ would have its outgoing arc to $j_\ell$ or
another client of smaller clock).  By repeating this argument, we get that  the
probability that all the arcs $(i_1, j_1),(i_2, j_2), \ldots, (i_k, j_k)$ exist in $\cgraph{}$ is at most
$$
\prod_{\ell=1}^k\frac{1}{|C(i_1, i_2, \ldots, i_\ell)|}.
$$

We then use the randomness of the facility-clocks to bound the probability that all the arcs $(j_0, i_1),$
$(j_1, i_2),$ $\ldots, (j_k, i_{k+1})$ exist. Similarly to above, we have that these arcs exist only
if
$$
Q_{i_{\ell+1}} =  \min \{Q_i \mid i\in F(j_0, j_1, \ldots, j_\ell)\} \qquad \mbox{ for } \ell = 0,
\ldots, k.
$$
As the clock $Q_i$ of a facility is distributed according to the  rate $o_i=y_i$, we have that
$\Pr[Q_{i_{\ell+1}} =  \min \{Q_i \mid i\in F(j_0, j_1, \ldots, j_\ell)\}$ equals
$$
\frac{y_{i_{\ell+1}}}{y(F(j_0,j_1, \ldots ,j_\ell)) }
$$
and again by using the
memorylessness property, all these arcs exist with probability at most
$$
 \frac{y_{i_1}}{y(F(j_0))} \cdot \frac{y_{i_2}}{ y(F(j_0,
  j_1)) }\cdots \frac{y_{i_{k+1}}}{y(F(j_0,j_1, \ldots ,j_k)) }.
$$

\noindent The bound  now follows since the  client-clocks and facility-clocks are independent.

Let us now calculate an upper bound on $\Pr[\langle j_0,i_1, j_1, i_2, \ldots ,i_k, j_k\rangle
\in \prefix(\cpath{j_0}{}) ]$. Similarly to above we have that the probability that all the arcs   $(i_1, j_1), (i_2,
j_2), \ldots, (i_k,j_k)$ exist in $\cgraph{}$ is at most~\eqref{eq:facedge}. We now bound the probability
that all the arcs $(j_0, i_1), (j_1, i_2), \ldots, (j_{k-1},i_k)$ exist in $\cgraph{}$. By using
analogous arguments to above, this is at most
$$
 \frac{y_{i_1}}{y(F(j_0))} \cdot\frac{y_{i_2}}{ y(F(j_0,
  j_1))} \cdots \frac{y_{i_k}}{y(F(j_0,j_1, \ldots ,j_{k-1}))}
$$
and the statement again follows from the independence of the facility-clocks and client-clocks.

\end{proof}

\subsubsection{Expected number of connection paths traversing an edge}

We now bound the expected number of connection paths that traverse an edge in the support
graph. When we say that a path visits $k$ clients before going through arc $(j,i)$ we mean that it
visits $k$ clients different from $j$ before going through the arc. 
\begin{lemma} 
\label{lem:marginal}
Consider a facility $i\in F$ and a client $j\in C$.
For any integer $k\geq 1$, the expected number of connection paths that visits $k$ clients before going through
arc $(i,j)$ is at most
$$
\frac{x_{ij}}{2^{\max(0,k-2)}}
$$
and, for any integer $k\geq 0$, the expected number of connection paths that visits $k$ clients
before going through the arc $(j,i)$ is at most
$$
 \frac{x_{ij}}{2^{\max(0,k-1)}}.
$$
\end{lemma}

\begin{proof}
We divide the proof into the following  cases: $k=0, k=1,$ and $k\geq 2$.

\

\textit{Case $k=0$.} In this case there is no path that visits $0$ clients before going through the arc
$(i,j)$. The expected number of paths that visits $0$ clients before going through the arc $(j,i)$ is equal to the probability of
that $\langle j,i \rangle \in \prefix(\cpath{j}{})$ which by Lemma~\ref{lem:pathprob} is at most
$y_i/y(F(j)) = y_i = x_{ij}$.

\ 

\textit{Case $k=1$.}  Note that any prefix of a connection path that visits $1$ client before going through the arc $(i,j)$
must be of the form $\langle j_0, i_1, j_1\rangle$ where $i_1 =i, j_1 =j$ and $j_0 \in
C(i_1)$. Hence, by linearity of expectation, we have that the expected number of such paths is at
most
\begin{align*}
\sum_{j_0 \in C(i_1)} \Pr[\langle j_0, i_1, j_1 \rangle \in \prefix(\cpath{j_0}{})] \leq \sum_{j_0 \in C(i_1)}
\frac{1}{|C(i_1)|} \frac{y_{i_1}}{y(F(j_0))}  = y_{i_1} = x_{i_1 j_1},
\end{align*}
where the inequality follows from Lemma~\ref{lem:pathprob} and the equalities follow from
$y(F(j_0)) = \sum_{i\in F} x_{ij_0} = 1$ and $x_{i_1j_1} = y_{i_1}$ since
$i_1\in F(j_1)$.

Let us now consider the expected number of connection paths whose prefix visits $1$ client before going through the arc
$(j,i)$. If we let $j_1 = j$ and $i_2 = i$, any such prefix has the form $\langle j_0, i_1, j_1,
i_2\rangle$ where $i_1 \in F(j_1)$ and $j_0 \in C(i_1)$. Hence, again by linearity of expectation
and by Lemma~\ref{lem:pathprob}, we have the upper bound of
\begin{align*}
\sum_{i_1 \in F(j_1)}\sum_{j_0 \in C(i_1)} \Pr[\langle j_0, i_1, j_1, i_2 \rangle \in \prefix(\cpath{j_0}{})] &\leq \sum_{i_1
  \in F(j_1)}\sum_{j_0 \in C(i_1)} \frac{1}{|C(i_1)|} \frac{y_{i_1}}{y(F(j_0))} \frac{y_{i_2}}{y(F(j_0,j_1))} \\
& \leq y_{i_2} \sum_{i_1
  \in F(j_1)}\sum_{j_0 \in C(i_1)} \frac{1}{|C(i_1)|} \frac{y_{i_1}}{y(F(j_1))} = y_{i_2} = x_{i_2 j_1},
\end{align*}
where the last inequality follows from $y(F(j_0,j_1)) \geq y(F(j_1))$ and $y(F(j_0)) = 1$.

\

\textit{Case $k\geq 2$:} We start by analyzing the expected number of connection paths that have
prefixes that visit
$k$ clients before going through arc $(i,j)$. For notational convenience let $i_k =
i$ and $j_k = j$. Any such prefix with nonzero probability has the form
$$
\langle j_0, i_1, j_{1}, i_{2}, \ldots, j_{k-1}, i_k, j_k \rangle,
$$
where $j_\ell \in C(i_{\ell+1})$ for $\ell = 0, 1, \ldots, k-1$ and $i_\ell \in F(j_{\ell}) $ for
$\ell = 1, \ldots, k$. Moreover, we have that $i_\ell \not \in F(j_{\ell+1})$, i.e., $i_\ell \in
F(j_\ell) \setminus F(j_{\ell+1})$. Indeed, as the client-clocks are decreasing along a path, if $i_\ell\in F(j_{\ell+1})$
 then $i_\ell$ would have its outgoing
arc to $j_{\ell+1}$ (or a client of even smaller clock) instead of $j_\ell$.   By linearity of
expectation, we can thus upper bound the expected number of such prefixes by the following sum:
$$
\sum_{j_{k-1} \in C(i_k)} \sum_{\substack{i_{k-1} \in F(j_{k-1}) \setminus F(j_k)\\ j_{k-2} \in C(i_{k-1})}}
\cdots \sum_{\substack{i_{1} \in F(j_{1}) \setminus F(j_{2}) \\ j_0\in C(i_1)}} \Pr[\langle j_0,i_1, j_1, i_2, 
\ldots, j_{k-1}, i_k, j_k \rangle \in \prefix(\cpath{j_0}{})].
$$
By Lemma~\ref{lem:pathprob} we have that  $\Pr[\langle j_0,i_1, j_1, i_2, j_2,
\ldots, j_{k-1}, i_k, j_k \rangle \in \prefix(\cpath{j_0}{})]$ is at most
\begin{align*}
&\prod_{\ell=1}^k\frac{1}{|C(i_1, i_2, \ldots, i_\ell)|} \cdot  \prod_{\ell=0}^{k-1}\frac{y_{i_{\ell+1}}}{y(F(j_0,
    j_1, \ldots, j_\ell))}  \leq  \prod_{\ell=1}^k\frac{1}{|C(i_\ell)|} \cdot\frac{y_{i_k}}{y(F(j_0))}\cdot  \prod_{\ell=0}^{k-2}\frac{y_{i_{\ell+1}}}{y(F(j_\ell,
    j_{\ell+1}))}  \\
&\leq \frac{y_{i_k}}{|C(i_k)| y(F(j_0))} \frac{y_{i_{k-1}}}{|C(i_{k-1})| y(F(j_0,
    j_1))}\prod_{\ell=1}^{k-2}\frac{y_{i_\ell}}{|C(i_\ell)|y(F(j_\ell, j_{\ell+1}))} \\
 &\leq \frac{y_{i_k}}{|C(i_k)|} \frac{y_{i_{k-1}}}{|C(i_{k-1})| }\prod_{\ell=1}^{k-2}\frac{y_{i_\ell}}{|C(i_\ell)|y(F(j_\ell, j_{\ell+1}))}.
\end{align*}
Substituting in this bound, we get that the expected number of connection paths is at most 
$$
\sum_{j_{k-1} \in C(i_k)} \sum_{\substack{i_{k-1} \in F(j_{k-1}) \setminus F(j_k)\\ j_{k-2} \in C(i_{k-1})}}
\cdots \sum_{\substack{i_{1} \in F(j_{1}) \setminus F(j_{2}) \\ j_0\in C(i_1)}}  \frac{y_{i_k}}{|C(i_k)|} \frac{y_{i_{k-1}}}{|C(i_{k-1})| }\prod_{\ell=1}^{k-2}\frac{y_{i_\ell}}{|C(i_\ell)|y(F(j_\ell, j_{\ell+1}))},
$$
which, by rearranging terms, equals
$$
\sum_{j_{k-1} \in C(i_k)}  \frac{y_{i_k}}{|C(i_k)|} \left(\sum_{\substack{i_{k-1} \in F(j_{k-1}) \setminus
    F(j_k)\\ j_{k-2} \in C(i_{k-1})}} \frac{y_{i_{k-1}}}{|C(i_{k-1})|}
\cdots \left(\sum_{\substack{i_{1} \in F(j_{1}) \setminus F(j_{2}) \\ j_0\in C(i_1)}} 
\frac{y_{i_1}}{|C(i_1)| y(F(j_{1}, j_{2}))}\right)\right).
$$
To analyze this expression, let us first consider the last term 
$$\sum_{i_{1} \in F(j_{1}) \setminus
  F(j_{2}) } \sum_{j_0\in C(i_1)} \frac{1}{|C(i_1)|} \frac{y_{i_1}}{y(F(j_1, j_2))} = \sum_{i_{1} \in
  F(j_{1}) \setminus F(j_{2}) } \frac{y_{i_1}}{y(F(j_1, j_2))}= \frac{y(F(j_{1}) \setminus
  F(j_{2}))}{y(F(j_1, j_2))}.$$ 
Recall that $F(j_{1})$ and $F(j_{2})$ are two sets such that $y(F(j_1)) = y(F(j_2)) = 1$.
 Let  $s = y(F(j_{1}) \cap F(j_{2}))$. Then
\begin{align*}
\frac{y(F(j_{1}) \setminus F(j_{2}))}{ y(F(j_{1},
  j_{2}))} = \frac{1-s}{2-s} \leq \frac{1}{2}.
\end{align*}
In general, we have for $\ell=1, \dots k-2$ that the  $\ell$-th last term is bounded by
$$
\frac{y(F(j_{\ell}) \setminus F(j_{\ell+1}))}{ y(F(j_{\ell},
  j_{\ell+1}))}  \leq \frac{1}{2}.
$$
Thus, repeating the same arguments for the $k-2$ last terms allow us to upper bound the expected number of
paths by
$$ \frac{1}{2^{\max(0,k-2)}} \sum_{j_{k-1} \in C(i_k)}  \frac{y_{i_k}}{|C(i_k)|} \left(\sum_{\substack{i_{k-1} \in F(j_{k-1}) \setminus
    F(j_k)\\ j_{k-2} \in C(i_{k-1})}} \frac{y_{i_{k-1}}}{|C(i_{k-1})|} \right)$$
and using that $y(F(j_{k-1}) \setminus F(j_k)) \leq 1$ this is at most
$y_{i_k}\frac{1}{2^{\max(0,k-2)}} = x_{i_kj_k} \frac{1}{2^{\max(0,k-2)}}$  as required.

Let us now bound the expected number of connections paths whose prefix visits $k$ clients before going through the arc
$(j,i)$. For notational convenience, let now $j_k= j$ and $i_{k+1} = i$. By the same arguments as
above, any such prefix with nonzero probability has the form $
\langle j_0, i_1, j_{1}, i_{2}, \ldots, j_{k-1}, i_k, j_k, i_{k+1} \rangle,
$
where $j_\ell \in C(i_{\ell+1})$ for $\ell = 0, 1, \ldots, k$ and $i_\ell \in F(j_{\ell}) \setminus F(j_{\ell+1}) $ for
$\ell = 1, \ldots, k-1$. By linearity of expectation, we can thus bound the expected number of
such prefixes by
$$
\sum_{\substack{i_k \in F(j_k) \\ j_{k-1} \in C(i_k)}} \sum_{\substack{i_{k-1} \in F(j_{k-1}) \setminus F(j_k)\\ j_{k-2} \in C(i_{k-1})}}
\cdots \sum_{\substack{i_{1} \in F(j_{1}) \setminus F(j_{2}) \\ j_0\in C(i_1)}} \Pr[\langle j_0,i_1, j_1,
\ldots, j_{k-1}, i_k, j_k, i_{k+1} \rangle \in \prefix(\cpath{j_0}{})].
$$
Similar to above, by applying Lemma~\ref{lem:pathprob} and rearranging the terms, we have the
following upper bound
$$
y_{i_{k+1}}\sum_{\substack{i_k \in F(j_k)\\j_{k-1} \in C(i_k)}} \frac{y_{i_k}}{|C(i_k)|}\left(\sum_{\substack{i_{k-1} \in F(j_{k-1}) \setminus F(j_k)\\ j_{k-2} \in C(i_{k-1})}}
\frac{y_{i_{k-1}}}{|C(i_{k-1})| y(F(j_{k-1}, j_k))}
\cdots \left(\sum_{\substack{i_{1} \in F(j_{1}) \setminus F(j_{2}) \\ j_0\in C(i_1)}} \frac{y_{i_1}}{|C(i_1)|y(F(j_1, j_2))}\right)\right).
$$
Again we have that the last term is at most $1/2$ and repeating that argument now for the last $k-1$
terms  allows us to upper bound the expected number of
connection paths by
$$ \frac{1}{2^{\max(0,k-1)}} y_{i_{k+1}} \sum_{\substack{i_k \in F(j_k) \\j_{k-1} \in C(i_k)}} \frac{y_{i_k}}{|C(i_k)|
  }.$$
As $y(F(j_k)) = 1$, this equals  $\frac{1}{2^{\max(0,k-1)}} y_{i_{k+1}} = \frac{1}{2^{\max(0,k-1)}} x_{i_{k+1}j_k}$ as required.

\end{proof}

The following corollary follows from the fact that  the expected number of connection paths that
traverse arc $(i,j)$ is at most $ \sum_{k=1}^\infty 
\frac{x_{ij}}{2^{\max(0,k-2)}} =3x_{ij}$ and the expected number of connection paths that traverse arc $(j,i)$
is at most
$\sum_{k=0}^\infty   \frac{x_{ij}}{2^{\max(0,k-1)}} = 3x_{ij}$.

\begin{corollary}
\label{cor:marginal}
The expected number of connections paths that traverse  the arc $(i,j)$ (and respectively  arc
$(j,i)$) is at most $ 3x^t_{ij}$. Hence, the  expected number of connection paths that traverse the
edge $\{i,j\}$ in any direction is at most $6x^t_{ij}$.
\end{corollary}

\subsection{Bounding the switching cost}
\label{sec:switchcost}
In this section, we prove \eqref{eq:switchcost} of Lemma~\ref{lem:inequalities}. Lemma~\ref{lem:switch1} shows \eqref{eq:switchcost} for a special case where $|Z^t|=1$.

\begin{lemma}\label{lem:switch1}
For some client $k\in C$, suppose that $(x^A,y^A), (x^B,y^B)\in\pfl$ satisfy $x^A_{ ij}=x^B_{ ij}$ for all $i\in F$ and $j\in C\setminus\{k\}$. If we use a single set of exponential clocks to construct both of their corresponding connection graphs, the expected number of clients whose connection paths are different is at most 7.
\end{lemma}
\begin{proof}
We first compare \cgraph{x^A} and \cgraph{x^B} to characterize their difference. Recall that every node has exactly one outgoing arc in a connection graph. For any $j\in C\setminus\{k\}$, its neighborhood in \sgraph{x^A} and in \sgraph{x^B} are the same; hence, the unique outgoing arc from $j$ is the same in both connection graphs (note that we use a single set of exponential clocks to define both connection graphs). The only client whose outgoing arc can be different in \cgraph{x^A} and \cgraph{x^B} is $k$.

Now suppose that a facility $ i\in  F$ has different outgoing arcs in \cgraph{x^A} and \cgraph{x^B}: $( i,j^A)\in \cgraph{x^A}$ and $( i,j^B)\in \cgraph{x^B}$ for $j^A\neq j^B$. This implies $\{ i,j^A\}\in \sgraph{x^A}$ and $\{ i,j^B\}\in \sgraph{x^B}$. We have that at least one of these two edges is absent in the other support graph, since otherwise both edges are in both support graphs and the choice of outgoing arc from $ i$ should have been consistent in both connection graphs. Suppose that $\{ i,j^A\}\notin \sgraph{x^B}$; in this case, $j^A=k$, since $\{ i,j^A\}\in \sgraph{x^A}$ and $k$ is the only client whose neighborhood can be different in the two support graphs. If $\{ i,j^B\}\notin \sgraph{x^A}$, $j^B=k$. In sum, if a facility $  i\in F$ has different outgoing arcs in \cgraph{x^A} and \cgraph{x^B}, one of the two outgoing arcs is towards $k$.

This characterization leads to the following claim:
\begin{claim}
For a client $j\in C$, if \cpath{j}{x^A} is different from \cpath{j}{x^B}, at least one of them contains $k$.
\end{claim}
\begin{proof}
Let $v$ be the last vertex of the maximal common prefix of \cpath{j}{x^A} and \cpath{j}{x^B}; the outgoing arcs of $v$ are different in \cgraph{x^A} and \cgraph{x^B} from the choice. If $v$ is a client, $v=k$ and hence the claim follows. If $v$ is a facility, one of its two outgoing arcs in \cgraph{x^A} and \cgraph{x^B} is towards $k$. Assume without loss of generality that $(v,k)\in\cgraph{x^A}$. In this case, either $(v,k)\in\cpath{j}{x^A}$, or $(v,k)\notin\cpath{j}{x^A}$ because $k$ was already visited by $\cpath{j}{x^A}$. In both cases, $k\in\cpath{j}{x^A}$.
\end{proof}

The connection path of $k$ automatically contains $k$; for any other client $j\in C\setminus\{k\}$, its connection path contains $k$ if and only if it contains an arc $( i,k)$ for some $ i\in  F$. Therefore, the lemma follows from the following claim (and its analogue for \cgraph{x^B}).
\begin{claim}
The expected number of connection paths in \cgraph{x^A} that contain $( i,k)$ for some $ i\in  F$ is at most 3.
\end{claim}
\begin{proof}
For each $i\in F$, the expected number of connection paths that contain $(i,k)$ is at most $3x^A_{ik}$ from Corollary~\ref{cor:marginal}. Thus, the expected number of connection paths that contain $(i,k)$ for \emph{any} $i\in F$ is at most $\sum_{i\in F}3x^A_{ik}=3\sum_{i\in F}x^A_{ik}=3$, since $(x^A,y^A)\in \pfl$.
\end{proof}
For some $j\in C$, if at least one of $\cpath{j}{x^A}$ and $\cpath{j}{x^B}$ contains $k$, one of the following is true: $j=k$ (there is one such connection path), $(i,k)\in\cpath{j}{x^A}$ for some $i\in F$ (in expectation, there are at most three such $j$), or $(i,k)\in\cpath{j}{x^B}$ for some $i\in F$ (again, there are at most three such $j$ in expectation). Thus, the expected number of clients whose connection paths are different is at most $1+3+3=7$.
\end{proof}

For the general case where $|Z^t|>1$, Corollary~\ref{cor:mulsw} shows \eqref{eq:switchcost} by applying Lemma~\ref{lem:switch1} multiple times.

\begin{corollary}\label{cor:mulsw}
For some $K\subset C$, suppose that $(x^A,y^A), (x^B,y^B)\in\pfl$ satisfy $x^A_{ ij}=x^B_{ ij}$ for all $ i\in F$ and $j\in C\setminus K$. If we use a single set of exponential clocks to construct both of their corresponding connection graphs, the expected number of clients that are assigned to different facilities is at most $7|K|$.
\end{corollary}
\begin{proof}
Let us denote the elements of $K$ as $k_1,\ldots,k_{|K|}$. Note that Lemma~\ref{lem:switch1} considers two sets of connection variables that are different in the neighborhood of only \emph{one} client, whereas this Corollary considers the case where the connection variables are different around $|K|$ clients. In order to apply Lemma~\ref{lem:switch1}, we can construct a series of solutions $(x^0,y^0),(x^1,y^1),\ldots,(x^{|K|},y^{|K|})$, which starts with $(x^0,y^0)=(x^A,y^A)$ and gradually looks more similar to $(x^B,y^B)$ until it ends with $(x^{|K|},y^{|K|})=(x^B,y^B)$. In particular, we ensure that $x^{\ell-1}$ and $x^{\ell}$ are different only in the neighborhood of $k_{\ell}$. By applying Lemma~\ref{lem:switch1} on each consecutive pair of these solutions, we obtain that the expected number of clients whose connection paths are different in $\cgraph{x^A}$ and $\cgraph{x^B}$ is at most $7|K|$. Note that, if a client is assigned to different facilities, its connection paths have to be different.
\end{proof}

\paragraph{Acknowledgments.} We thank Claire Mathieu for inspiring discussions and the anonymous reviewers of the conference version of this paper for their helpful comments.

\bibliographystyle{abbrv}
\bibliography{lit}

\appendix

\section{Preprocessing of the LP soution}
\label{ap:pre}
In this appendix, we present the two preprocessings we apply to the LP solution.

\subsection{First preprocessing}

The first preprocessing formalized by Lemma~\ref{lem:prepro} is due to Eisenstat et
al.~\cite{EMS}. For the sake of completeness, we present their preprocessing in this section under our notation.

Let $(\bar x,\bar y, \bar z)$ be a given LP solution. We shall in polynomial time output a feasible solution
$(x,y,z)$ satisfying the following:
\begin{itemize}
\item The cost of $(x,y,z)$ is at most twice the cost of $(\bar x,\bar y,\bar z)$.
\item If we let $Z^t = \{j\in C \mid x^t_{ij} \neq
  x^{t+1}_{ij}  \mbox{ for some } i\in F\} $ denote the set of clients that changed its fractional connection between time step $t$
  and $t+1$, then
  \begin{equation}\label{e:pre1goal}
  \sum_{t=1}^{T-1} |Z^t| \leq \sum_{t=1}^{T-1}\sum_{i\in F, j\in C}z_{ij}^t.
  \end{equation}
\end{itemize}

Note that this will prove Lemma~\ref{lem:prepro}. The vector $x$ is constructed from
$\bar x$ as follows:

\begin{enumerate}

\item  For each client $j$, we decide on a set $L_j=\{t^j_0,t^j_1,\ldots,t^j_{\iota(j)}=T+1\}$ of \emph{boundary time steps} that divide the entire time span into \emph{time intervals} $[t^j_0,t^j_1)$, $\cdots$, $[t^j_{\iota(j)-1},t^j_{\iota(j)})$ :
\begin{itemize}
\item Let $t^j_0 =1$ and $\iota=1$.
\item Now select $t^j_\iota$ to be the  largest $t \in (t^j_{\iota-1}, T+1]$ such that
  $\sum_{i \in F}( \min_{t^j_{\iota-1} \leq u < t} \bar x^u_{ij} ) \geq 1/2$.
\item If $t^j_\iota = T+1$ we are done selecting boundaries: $\iota(j)\gets\iota$, and $L_j \gets\{t^j_0,\ldots,t^j_{\iota(j)}\}$.\\Otherwise, we increment $\iota$ and repeat the previous step to select the next boundary.
\end{itemize}
Note that each client defines its own division of the time span.

\item Once we have decided on these time intervals, we redefine the connection variables so that, for each client $j$, its connection variables do not change within each time interval defined by itself.\\
For each client $j$ and each of its time interval $[t^j_k, t^j_{k+1})$, we set
$$
x^t_{ij} : = \frac{\min_{t^j_{k} \leq u < t^j_{k+1}} \bar x^u_{ij}}{ \sum_{i' \in F}( \min_{t^j_{k} \leq u < t^j_{k+1}} \bar x^u_{i'j} )} \qquad \mbox{ for each } i\in F\mbox{ and } t\in [t^j_k, t^j_{k+1}).
$$

Note that the right-hand-side is not dependent on $t$, achieving the desired property.
\end{enumerate}

By construction we have $x^t_{ij} \leq \bar x^t_{ij}/ \sum_{i' \in F}( \min_{t_{k}^j \leq u <
    t_{k+1}^j} \bar x^u_{i'j} ) \leq 2 \bar x^t_{ij}$. Therefore the connection cost
  of $x$ is at most twice the connection cost of $\bar x$. Moreover, if we let $y =
  2\bar y$, then $x_{ij}^t \leq y_i^t$ for all $i\in F, j\in C$ and $t\in [T]$. Finally,  by
  construction $\sum_{i\in F} x^t_{ij} =1$ for all $j\in C$ and $t\in [T]$, and hence we can conclude
  that
$$
(x^t, y^t) \in \pfl \qquad \mbox{ for all } t\in[T].
$$

We finish the description of the preprocessing by specifying $z$. Towards this aim we use the following fact from \cite{EMS}:
\begin{fact}
For each client $j \in C$ and its time interval $[t^j_{k-1}, t^j_{k})$ where $1\leq k<\iota(j)$,
\begin{align*}
\sum_{t^j_{k-1}\leq t < t^j_{k}} \sum_{i \in F} \bar z_{ij}^t > 1/2.
\end{align*}
\end{fact}

Note that this implies that $\sum_{j\in C} (\iota(j)-1)\leq \sum_{t=1}^{T-1}\sum_{i\in F, j\in C}2 \bar z_{ij}^t$. On the other hand, recall that the fractional connection of a client $j$ does not change within each of its time interval from constuction; thus, $\sum_{t=1}^{T-1} |Z^t|\leq \sum_{j\in C} (\iota(j)-1)$. We have, therefore,\[
\sum_{t=1}^{T-1} |Z^t| \leq \sum_{t=1}^{T-1}\sum_{i\in F, j\in C}2 \bar z_{ij}^t.
\]Thus, if we can obtain (feasible) $z$ such that $\sum_{t=1}^{T-1}\sum_{i\in F, j\in C} z_{ij}^t = \sum_{t=1}^{T-1}\sum_{i\in F, j\in C}2 \bar z_{ij}^t$, \eqref{e:pre1goal} would follow, and the switching cost would be at most twice the original switching cost.

We first set $z^t_{ij} := x_{ij}^t - x_{ij}^{t+1}$ for all $i\in F, j\in C$ and $t\in [T)$, and we have $
\sum_{t=1}^{T-1}\sum_{i\in F, j\in C} z_{ij}^t \leq \sum_{j\in C} (\iota(j)-1)
$ from construction, which in turn implies $
\sum_{t=1}^{T-1}\sum_{i\in F, j\in C} z_{ij}^t \leq \sum_{t=1}^{T-1}\sum_{i\in F, j\in C}2 \bar z_{ij}^t
$. Then we can increase $z$ in an arbitrary manner until we achieve $\sum_{t=1}^{T-1}\sum_{i\in F, j\in C} z_{ij}^t = \sum_{t=1}^{T-1}\sum_{i\in F, j\in C}2 \bar z_{ij}^t$.


\subsection{Second preprocessing}

The second preprocessing formalized by Observation~\ref{obs:prepro} follows from applying the standard technique of duplicating facilities. We present its details in this section.

Let us first review the standard technique that is used by multiple algorithms for the classic problem: it ensures $x_{ij}\in\{0,y_i\}$. Suppose there exists a facility $i$ that is open by the fraction of $.6$ and connected to clients $a$, $b$, $c$, and $d$ each by $.1$, $.4$, $.4$, and $.6$, respectively, in the LP solution. The standard technique in this case duplicates $i$ into three copies, each of which is to be opened by $.1$, $.3$, and $.2$. Then $a$, $b$, $c$, and $d$ is respectively connected to the first one, two, two, and three copies of the facility. Figure~\ref{f:ap2a} illustrates this duplication. Note that this technique modifies the problem instance since it creates copies of facilities. However, we define the metric so that the copied facilities are exactly at the same position as the original facility: e.g., $d(i,j)=d(i_1,j)$ for all $j\in C$, and the copied facilities are also defined to have the same opening cost as the original. Thus, this modification does not change the cost of the LP solution, and if we find an (approximate) solution to the new instance, it translates back to the original instance by opening the original facilities instead of their duplicates. In general, the connection variables of each facility determine the set of ``threshold'' values (\{.1,\, .4,\, .6\} in this case) to be used to split that facility, and the facility is split into multiple copies each of which is to be opened by the fraction equal to the difference of two consecutive threshold values ($.1$, $.4-.1$, and $.6-.4$ in this case).

\begin{figure}[ht]
\begin{center}
\includegraphics[width=0.9\textwidth]{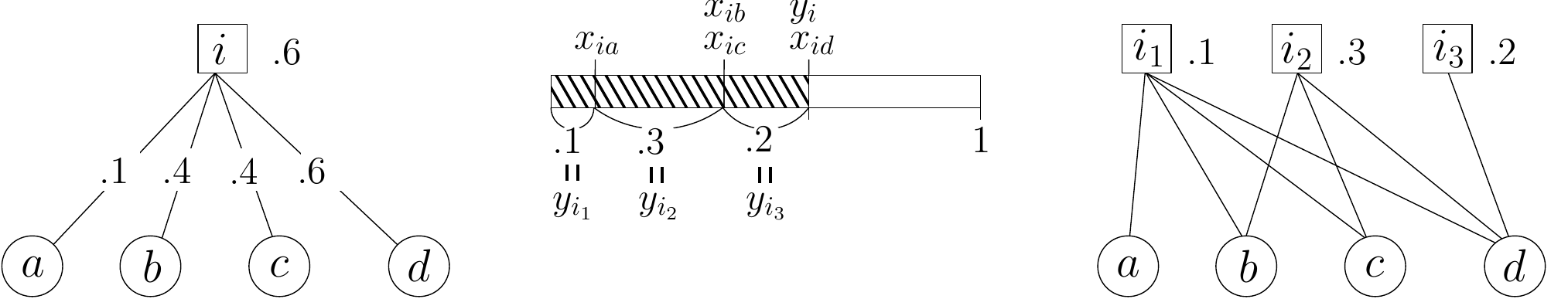}
\end{center}
\caption{The standard technique applied to the classic problem. (Part of) an LP solution is shown on the left. Numbers on the edges show the connection variables; next to the facilities (represented as squares) are the opening variables. Result of the preprocessing is shown on the right: connection variables are omitted, as they are equal to the incident opening variables.}
\label{f:ap2a}
\end{figure}

For the dynamic problem, we additionally need to ensure that each copy of a facility is open by the same fraction (or by zero) at every time step (Property~\ref{enum:prepro2:2} of Observation~\ref{obs:prepro}). In order to obtain this property, every time step will share a single set of threshold values: we determine the thresholds for facility $i\in F$ by taking all its LP variables across all the time steps: $\{x_{ij}^t\mid j\in C,t\in [T]\}\cup\{y_i^t\mid t\in [T]\}$. Figure~\ref{f:ap2b} shows an example: the set of threshold values is $\{.1,.3,.4,.7,.8\}$, and therefore we create five copies of the facility, each of which is to be opened respectively by $.1$, $.2$, $.1$, $.3$, and $.1$. Now the openings and connections of the facility are implemented by subsets of the first few copies of the facility.

\begin{figure}[ht]
\begin{center}
\includegraphics[width=0.9\textwidth]{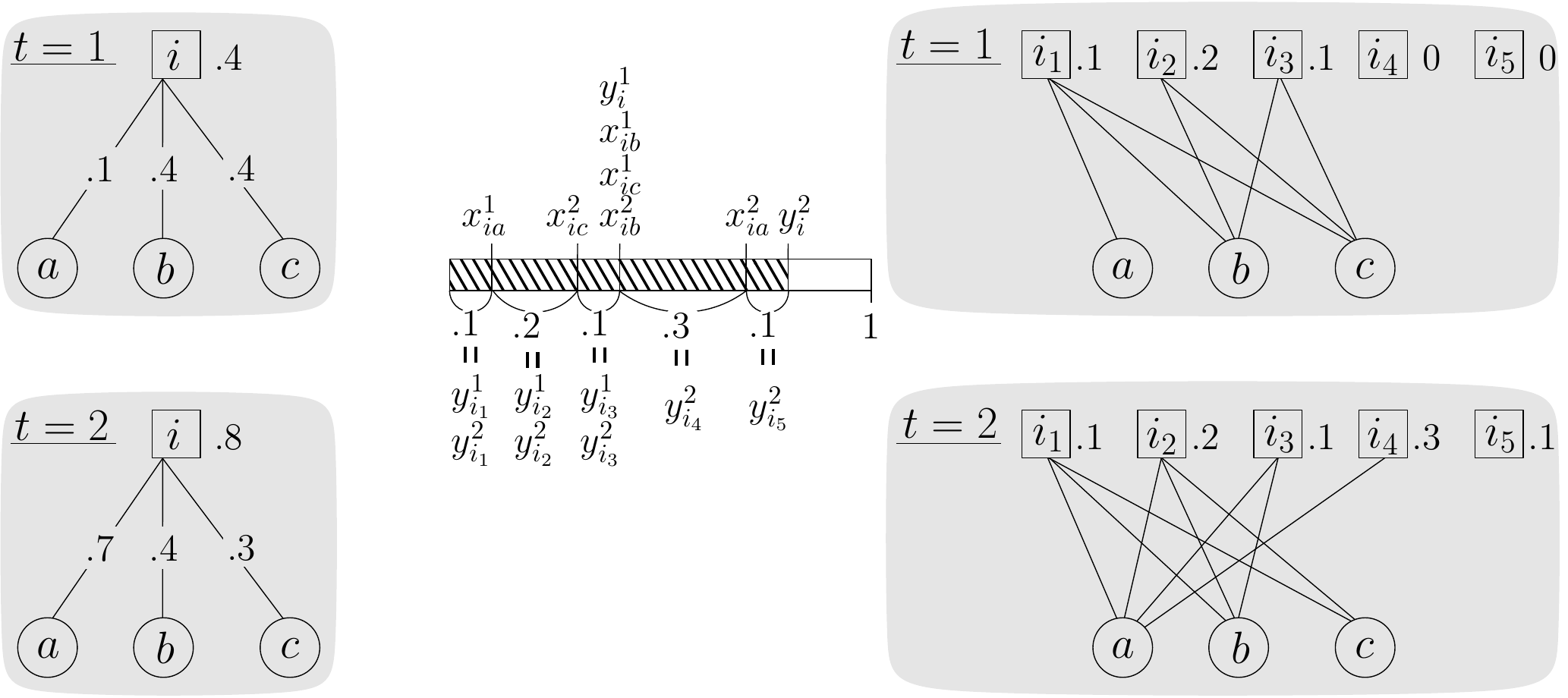}
\end{center}
\caption{The second preprocessing. $T=2$.}
\label{f:ap2b}
\end{figure}

Finally, note that the connection variables of the same value are implemented by connecting to the same set of copies. In Figure~\ref{f:ap2b}, $x_{ib}^1=x_{ib}^2=.4$ and both are implemented by connecting to $\{i_1,i_2,i_3\}$ in their respective time steps. Thus, if a client $j$ has exactly the same set of connection variables in two time steps, this remains the case even after the preprocessing.

\end{document}